\documentclass{acm_proc_article-sp}
\usepackage{amssymb}
\usepackage{amsmath}
\usepackage{graphicx}

\newcommand\be{\begin{equation}} 
\newcommand\ee{\end{equation}} 
\newcommand\belabel{\begin{equation}}
\newcommand\eelabel{\end{equation}}
\newcommand\xsi{X_i(S_i)}

\newcommand\xhi{X_{hi}(S_i)}

\newcommand\kk{\mathbf{k}}
\newcommand\XX{\mathbf{X}}
\newcommand\SSS{\mathbf{S}}
\newcommand\SX{\SSS,\XX}
\newcommand\rksx{r(\kk, \SSS, \XX)}

\newcommand\ri{R_i}
\newcommand\riest{\hat{R}_i}
\newcommand\diest{\hat{D}_i}

\begin{document}

\title{Portfolio Allocation for Sellers in Online Advertising}

\numberofauthors{4} 
%
\author{
%
%
\alignauthor
Ragavendran Gopalakrishnan\\
      \affaddr{Xerox Research}\\
      \email{ragad3@caltech.edu}
\alignauthor
Eric Bax\\
      \affaddr{Yahoo}\\
       \email{ebax@yahoo-inc.com}
\alignauthor
Krishna Prasad Chitrapura\\
      \affaddr{Qikwell Technologies}\\
       \email{kp@qikwell.com}
\and
\alignauthor Sachin Garg\\
      \affaddr{American Express}
      \email{sachin.garg@aexp.com}
}

\maketitle

\begin{abstract}
In markets for online advertising, some advertisers pay only when users respond to ads. So publishers estimate ad response rates and multiply by advertiser bids to estimate expected revenue for showing ads. Since these estimates may be inaccurate, the publisher risks not selecting the ad for each ad call that would maximize revenue. The variance of revenue can be decomposed into two components -- variance due to `uncertainty'  because the true response rate is unknown, and variance due to `randomness' because realized response statistics fluctuate around the true response rate. Over a sequence of many ad calls, the variance due to randomness nearly vanishes due to the law of large numbers. However, the variance due to uncertainty doesn't diminish. 

We introduce a technique for ad selection that augments existing estimation and explore-exploit methods. The technique uses methods from portfolio optimization to produce a distribution over ads rather than selecting the single ad that maximizes estimated expected revenue. Over a sequence of similar ad calls, ads are selected according to the distribution. This approach decreases the effects of uncertainty and increases revenue. 
\end{abstract}

\section{Introduction}
In recent years, online advertising has emerged as a hugely profitable industry for publishers and advertisers alike on the Internet. Today's online advertiser faces two choices for placing their ads -- display advertising and search advertising. Display advertising is most similar to traditional advertising, in that ads are placed in popular websites alongside the published content. An example would be an airline ad on the travel section of a news website. Search advertising, on the other hand, involves ads being shown alongside search results on certain keywords, on search websites. An example would be an airline ad alongside the list of search results for the keyword `travel'. Online publishers (in both display and search advertising) want to maximize the revenue generated from selling real estate (`impressions' or `ad calls') on their websites, to advertisers.

In the traditional model, advertisers pay publishers just for showing their ads. This is the CPM (cost per mille) model, where advertisers incur a cost per thousand impressions for which their ads are selected. In this model, the entire advertising risk is borne by the advertisers, because the publishers are guaranteed revenue for showing ads. Alternative models have since emerged, that split the risk more evenly between the publisher and the advertiser. For example, in the CPC (cost per click) model, an advertiser pays only when a user clicks on their ad, causing their landing page to be displayed in the user's web browser. Clearly, the risk is split -- a low click volume is bad for the publisher, but minimizes the advertiser's costs, whereas a high click volume with poor conversion is great for the publisher, but bad for the advertiser due to low return on investment. The CPA (cost per action) model is the other extreme, where the entire risk is borne by the publisher -- an advertiser typically pays only when a user completes an action that guarantees a conversion. Examples include users navigating to a specified page on the advertiser's web site, filling out a web form that submits contact information to the advertiser (called lead generation), and completing an online purchase. For more about online advertising markets, refer to Varian \cite{varian06,varian09}, Edelman et al. \cite{edelman07}, and Lahaie and Pennock \cite{lahaie07}.

Advertisers submit bids for their ads to the publisher, indicating how much they are willing to pay. A publisher usually partitions available ad calls into different segments, called markets, and select ads to be shown for each market. When a publisher is faced with several ads to choose from, the revenue optimizing choice is most obvious if all ads are CPM ads -- there is a wide range of literature that deals with optimal auctions, refer to Riley and Samuelson \cite{riley81}, Myerson \cite{myerson81} and the textbook by Krishna \cite{krishna09}; most notable among these is the generalized second price auction, first analyzed by Edelman et al. \cite{edelman07} and Varian \cite{varian06}. If CPC and/or CPA ads are involved, then the function that is optimized by these auction mechanisms is the expected revenue, which is the advertiser payment times user response rate (commonly referred to as clickthrough rate (CTR) in the context of CPC ads). Since the response rates of the ads may not be known accurately, the expected revenue has to be estimated. We note that the mechanism used by the publisher to select ads affects the incentives of the competing advertisers, which in turn affects their bidding behavior. While we do not consider these effects in this paper, in \cite{vcg_portfolio}, Li et al. discuss a pricing mechanism for portfolio allocation, taking this into account.

Selecting a single ad that would result in the maximum estimated expected revenue and showing it on all ad calls in a market is risky for the publisher, because the actual revenue may be less than the estimated expected revenue. To our knowledge, most existing techniques for reducing this risk use samples from historical performance data for the ads to obtain improved response rate estimates, see Richardson et al. \cite{richardson07}, Agarwal et al. \cite{agarwal09c} and Graepel et al. \cite{graepel10}. Even though such methods hold much merit, there is an inherent difficulty involved in `learning' the response rates in this manner, especially when these response rates are very small, which is usually the case, see \cite{comscore08}.

In this paper, we outline a new approach for selecting ads that can augment existing techniques that involve intelligent learning of the response rates, to provide better risk management for the publisher. This approach involves selecting a portfolio of ads to share the ad calls in the market, so as to reduce the variance of revenue. We use techniques from portfolio optimization for this purpose. Portfolio optimization itself is not a new idea; it has existed in the world of finance for decades. For more on portfolio allocation techniques, refer to the text by Fabozzi \cite{fabozzi07}, or the papers by Markovitz \cite{markovitz52}, Lintner \cite{lintner65}, Sharpe \cite{sharpe64}, and Tobin \cite{tobin58}.

The variance in revenue is due to two factors which we refer to in this paper, as `undertainty' and `randomness'. Uncertainty refers to the fact the true response rates are not known and are estimated. Randomness refers to the fact that the response rate, by definition, is still an `average' -- it represents the probability of eliciting a response, and the realized response statistics will fluctuate around it. The variance due to randomness diminishes according to the law of large numbers, as the number of allocated ad calls increases (which happens over a large time period). The same cannot be said for the variance due to uncertainty. Portfolio allocation specifically targets the component of variance due to uncertainty by diversifying, i.e., spreading the ad calls over multiple ads. As a side-effect of reducing risk, our simulations show that there is potential for increasing the actual (realized) revenue, so even risk-neutral publishers benefit from this approach.

Section \ref{formal_model} establishes a formal model. Section \ref{optimal_allocation} describes how to optimize for a combination of estimated expectation and variance of revenue in our formal model. Section \ref{var_analysis} analyzes the components of variance due to uncertainty and randomness, and the impact of diversification. Section \ref{simulations} uses simulations to illustrate that using portfolio allocation has the potential to increase actual revenue. Section \ref{discussion} concludes with a discussion of directions for future work.

\section{Formal Model} \label{formal_model}
Let $m$ be the number of ad calls and let $n$ be the number of ads in a market. An allocation vector $\kk = (k_1, \ldots, k_n)$ specifies the number of ad calls to allocate to each ad. The goal is to select an optimal allocation vector $\kk^*$ that mediates a tradeoff between maximizing expected revenue and minimizing variance of revenue.

Assume each ad $i$ in the market is generated by a distribution over possible ads, each of which has a response rate. Let $S_i$ be the random variable that denotes the response rate of ad i. Let $R_i$ be the distribution of $S_i$. Let $X_i(S_i)$ be the random variable that denotes the revenue from showing ad $i$ on an ad call. For example, if an advertiser pays his bid $b$ only when the ad elicits a response, then

\be
X_i(S_i) = \left\{ \begin{array}{c} \hbox{$b$ with probability $S_i$} \\ \hbox{$0$ with probability $1-S_i$} \\ \end{array} \right. .
\ee
\noindent

Define random variables $\xhi$, for $h$ in $\{1, \ldots , m\}$ and $i$ in $\{1, \ldots, n\}$, to be the revenue $\xsi$ if ad call $h$ is allocated to ad $i$. (Random variables $\xhi$ are independent copies of $\xsi$.) Response rate $S_i$ is drawn once (according to $R_i$) and determines a distribution for revenue for all copies of $\xsi$: $X_{1i}, \ldots, X_{mi}$. But $\xhi$ is redrawn i.i.d. according to that distribution for each ad call $h$. Think of it as drawing a coin from a bag of coins to determine the response probability $S_i$ for ad $i$, then tossing that same coin once for each ad call to determine the values $X_{1i}, \ldots, X_{mi}$. In terms of this notation, `uncertainty' is tied to the distribution $R_i$ that generates the response rates, whereas `randomness' is tied to the distribution $\xsi$ that determines the revenue for a given response rate.

The goal is to select an allocation $\kk^*$ to maximize expected return subject to controls on variance of returns. The expectation and variance are over $\SSS = (S_1, \ldots, S_n)$ and $\XX = (X_{11}, \ldots , X_{1n}, \ldots , X_{m1}, \ldots , X_{mn})$. However, $\SSS$ and $\XX$ are unknown at the time of portfolio allocation, so their statistics must be estimated.

In the sections that follow, expectations, variances, and covariances are over the distributions of the random variables in subscripts. For example, $E_{\SSS}$ is expectation over the distribution of $\SSS$. Similarly, $Var_{\SX}$ is variance over the joint distribution of $(\SX)$.

\section{Estimated Optimal Allocation} \label{optimal_allocation}
Let $\kk$ be an allocation. Assume, without loss of generality, that the first $k_1$ ad calls are allocated to ad 1, the next $k_2$ are allocated to ad 2, and so on. Then the revenue for allocation $\kk$ is

\be
\rksx = \sum_{i=1}^{n} \sum_{h=k_1 + \ldots + k_{i-1}+1}^{k_1 + \ldots + k_i} \xhi.
\ee
\noindent
So the allocation optimization problem (AOP) is:

\be
\max_{\kk} E_{\SX} \rksx
\ee
\noindent
subject to

\be
Var_{\SX} \rksx \leq d,
\ee
\noindent
where $d$ is a specified bound on variance, and

\be
\forall i: k_i \geq 0, \hbox{ and } \sum_{i=1}^{n} k_i = m.
\ee
\noindent
Since expectations are linear, the expected revenue is

\belabel
E_{\SX} \rksx = \sum_{i=1}^{n} k_i E_{S_i, X_i} \xsi  \label{expectation}
\eelabel
\noindent
The variance of revenue is:

\be
Var_{\SX} \rksx 
\ee
\be
= \sum_{i=1}^{n} \sum_{j=1}^{n} k_i k_j Cov_{S_i,S_j}[E_{X_i} X_i(S_i), E_{X_j} X_j(S_j)] 
\ee
\belabel
+ \sum_{i=1}^{n} k_i E_{S_i} Var_{X_i} \xsi. \label{var}
\eelabel
\noindent
(See Appendix \ref{payoff_variance} for the proof.)

Define matrix $\mathbf{A}$ as

\be
a_{ij} = Cov_{S_i,S_j}[E_{X_i} X_i(S_i), E_{X_j} X_j(S_j)],
\ee
\noindent
 and define vectors $\mathbf{b}$ and $\mathbf{c}$:

 \be
 b_i = E_{S_i} [ Var_{X_i} \xsi ] \hbox{ and } c_i = E_{S_i, X_i} \xsi.
 \ee

(Please excuse the abuse of notation: the symbol $b$ represents an advertiser's bid elsewhere in this paper.) We now relax the constraint that $k_i$ can take only integral values. Then, the allocation optimization problem can be stated as

 \be
 \max_{\kk} \mathbf{c}^{T} \kk
 \ee

 subject to

 \be
 \kk^{T} \mathbf{A} \kk + \mathbf{b}^{T} \kk \leq d,
 \ee

 \be
 \kk \geq \mathbf{0} \hbox{ and } \mathbf{1}^T \kk = m.
 \ee

Call this the matrix allocation problem (MAP). This is a convex quadratic programming problem, which can be solved by any of a number of available quadratic programming (QP) solvers, employing techniques such as Wolfe's method \cite{wolfe59}, which is covered by Franklin \cite{franklin80} and other texts on linear and nonlinear programming.

An alternative formulation uses a parameter $q \in [0,\infty)$ to express how much to weight average returns versus variance. Solve the problem

\be
\min_{\kk}  \kk^{T} \mathbf{A} \kk + \mathbf{b}^{T} \kk - q \mathbf{c}^{T} \kk
\ee
\noindent
subject to

\be
 \kk \geq \mathbf{0} \hbox{ and } \mathbf{1}^T \kk = 1.
 \ee

Call this the $q$-weighted matrix allocation problem (QMAP). This convex quadratic programming problem is in a form that is convenient for many QP solvers. (For general background on allocation problems, refer to Franklin \cite{franklin80} or another text on mathematical programming and optimization.)

\section{Analysis of Variance of Revenue} \label{var_analysis}
The ad call allocation problems MAP and QMAP have the same form as the standard portfolio allocation problem in finance. In finance, an investor seeks to allocate funds among investments, with the goals of achieving high expected returns and low variance of returns. In online advertising, a risk-averse publisher seeks to allocate ad calls among ads, with the goals of achieving high expected revenue and low variance of revenue.

Like financial investors selecting a portfolio, using MAP and QMAP causes revenue-seeking, risk-averse publishers to:

\begin{enumerate}
\item Allocate more ad calls to ads that have higher expected revenue.
\item Allocate more ad calls to ads that have lower variance of revenue.
\item Diversify: spread ad calls over multiple ads.
\end{enumerate}

\subsection{Variance of revenue due to a single ad}
To begin, we focus on variance of returns due to a single ad $i$ being allocated to $k_i$ ad calls in the market. To focus on individual allocations, let us assume for now that covariances with respect to all $S_i$ and $S_j$ are zero: $\forall i \not= j, a_{ij}=0$. (This occurs when ad response rates are estimated independently.)

In a portfolio allocation, the portion of variance in revenue due to ad $i$ is then given by

\be
Var_i = k_i^2 Var_{S_i} [E_{X_i} X_i(S_i)] + k_i E_{S_i} [ Var_{X_i} X_i(S_i) ]. \label{var_i}
\ee

The first term is variance due to uncertainty. The second term is variance due to randomness. Variance due to uncertainty scales with the square of allocated ad calls $k_i$. Variance due to randomness scales linearly. Uncertainty increases with allocation size because the actual response rate $S_i$ is drawn once and applies to all ad calls allocated to ad $i$, making their revenues correlated. In contrast, deviations in revenue due to differences between average and realized response rates are independent from ad call to ad call.

We simplify the notation for the ensuing discussion as follows. Let

\begin{itemize}
\item $k$ (instead of $k_i$) count ad calls allocated to ad $i$,
\item $b$ be the advertiser's bid per response,
\item $p$ be the ad's (unknown) response rate, and
\item $\sigma$ be the standard deviation of error in estimating $p$.
\end{itemize}

Assuming an unbiased estimate of $p$, and assuming that the advertiser pays his bid when a user responds to his ad, Equation (\ref{var_i}):

\be
Var_i \approx k^2 b^2 \sigma^2 + k b^2 p (1-p), \label{var_i_approx}
\ee

Expected revenue per ad call is $bp$. Let $c$ be the expected revenue required for the ad's offer to be competitive. Then $b$ needs to be at least $\frac{c}{p}$. Substitute into Approximation (\ref{var_i_approx}):

\be
Var_i \approx \frac{k^2 c^2 \sigma^2}{p^2} + \frac{k c^2 (1-p)}{p}. \label{via_2}
\ee

As ads are shown to users and response data is collected, uncertainty about response rates decreases. Suppose an ad has actual response rate $p$ and obtains $u$ responses from being allocated to $v$ ad calls. Treating each ad call as a Bernoulli trial \cite{feller68} with success probability $p$, $\frac{u}{v}$ is an unbiased estimator of $p$, with standard deviation

\be
\sigma = \sqrt{\frac{p(1-p)}{v}}.
\ee

Substitute $\sigma= \sqrt{\frac{p(1-p)}{v}}$ into Approximation (\ref{via_2}):

\be
Var_i \approx \frac{k^2 c^2 (1-p)}{v p} + \frac{k c^2 (1-p)}{p}. \label{via_3}
\ee

Since the first term accounts for uncertainty and the second for randomness, the ratio of uncertainty to randomness is about $k:v$. For example, if the number of learning ad calls is about nine times the current-session allocation $k$, then uncertainty accounts for about $10\%$ of the variance in revenue. In general, an ad contributes more variance to revenue when its allocation $k$ is larger, when the response rate $p$ is smaller, and when fewer learning ad calls $v$ have been used to estimate the response rate.

\subsection{Effect of diversification on variance of revenue}
Suppose there are $r$ ads available, with independent and identical sets of distributions $S_i$ and $X_i(S_i)$. Then all allocations $(k_1, \ldots, k_r)$ of $k$ ad calls have the same expected revenue. Independence implies

\be
\forall i \not= j: k_i k_j Cov_{S_i,S_j}[E_{X_i} X_i(S_i), E_{X_j} X_j(S_j)]  = 0.
\ee

So the variance of revenue is

\be
\sum_{i=1}^{r} \left( k_i^2 Var_{S_i} [E_{X_i} X_i(S_i)] + k_i E_{S_i} [ Var_{X_i} X_i(S_i) ] \right) .
\ee

Let the uncertainty terms

\be
Var_{S_1} [E_{X_1} X_1(S_1)] = \ldots = Var_{S_r} [E_{X_r} X_r(S_r)] = \alpha.
\ee

Let the randomness terms

\be
E_{S_1} [Var_{X_1} X_1(S_1) ] = \ldots = E_{S_r} [ Var_{X_r} X_r(S_r) ] = \beta.
\ee

If all ad calls are allocated to a single ad, then the variance is $k^2 \alpha + k \beta$. If the ad calls are distributed uniformly over the ads, then the variance is $\frac{1}{r} k^2 \alpha + k \beta$. So diversification reduces variance caused by uncertainty.

\subsection{Covariance}
Uncertainty is completely correlated among ad calls allocated to the same ad, and it may also be correlated between different ads. In practice, uncertainty becomes correlated when ads ``share" learning. For example, in a tree-based model for learning response rates, the ad calls and responses for each ad may influence response rate estimates for other ads in the same branch of the tree. As a result, differences between estimated and actual expected response rates are likely to become correlated for ads that are neighbors in the tree.

Empirical data can be used to estimate covariance of expected returns among ads. A model for the covariance can be based on whether ads share a branch or sub-branch in a tree model for response rate estimation (see Agarwal et al. \cite{agarwal07}, Dudik et al. \cite{dudik07}, and Gelman and Hill \cite{gelman07}), are in the same cluster in a cluster model (see Regelson and Fain \cite{regelson06}), have similar scores for factors in a factor-based model (see Agarwal and Chen \cite{agarwal09b}, Weinberger et al. \cite{weinberger09}, and Richardson et al. \cite{richardson07}), or use the same rules in a rule-based model (see Dembczynski et al. \cite{dembczynski08}). The model for covariance can be trained by using converged response rate estimates for a population of experienced ads as proxies for actual response rates and observing how differences between early estimates and converged estimates are correlated for ads that ``share" learning.

\section{Simulations} \label{simulations}
This section uses simulations to show that using portfolio allocation to decrease estimated variance of revenue can cause an increase in the actual expected revenue. The simulations focus on markets for display advertising that have a mix of CPC and CPA ads. These results should also apply to markets that have only CPA ads, with a variety of definitions of an action and hence a variety of response rates. This is the typical case for CPA-dominated markets in display advertising.

In the simulations, ad response rates $S_i$ are independent of each other. For each ad $i$, let $\ri$ be the actual prior distribution for $S_i$. Let $\riest$ be the estimated prior distribution for $S_i$. Such a distribution can be based on historical response rates for ads. Using Bayesian statistics, let $\diest$ be the estimated posterior distribution for $S_i$. Statistics for this estimated posterior can be computed from the estimated prior $\riest$ and the past performance of ad $i$. (Refer to Berger \cite{berger85} or another text on Bayesian statistics for methods to compute statistics for the estimated posterior.)

Each simulation follows the steps:

\begin{enumerate}
\item Generate ad response rates $S_i$ at random based on actual priors $\ri$.
\item For each ad, randomly generate a series of 100,000 ``learning'' ad calls, with response rate $S_i$, and record the number of responses.
\item For each ad, compute an estimated posterior $\diest$ based on the ad's estimated prior $\riest$ and the number of responses to learning ad calls. (For the method to do this, refer to a textbook on Bayesian statistics, such as Berger \cite{berger85}.)
\item Use QMAP to allocate ad calls over the ads, based on statistics over the estimated posteriors $\diest$.
\item Record the actual expected revenue, $\sum_i k_i S_i b_i$ where $k_i$ is the number of ad calls allocated to ad $i$. This is the expected revenue achieved by the QMAP allocation.
\item For comparison, identify a``single winner'' ad -- an ad with maximum estimated revenue based on the estimated posteriors $\diest$. (The revenue estimate is the bid times the mean of the estimated posterior.) This is the ad that would be selected to receive all ad calls without portfolio allocation, based on the available information. Record its actual expected revenue.
\item Also for comparison, record the ideal expected revenue, $\max_i S_i b_i m$, where $b_i$ is the bid for ad $i$. This is the expected revenue if perfect knowledge of response rates could be used to select an ad with maximum actual expected revenue.
\end{enumerate}

Each simulation computes the QMAP allocation and collects results for all $q$ values in 0 to 1500 with increments of 25, then the values 1750, 2000, 3000, 4000, 5000, 7500, 10,000, 15,000, and 20,000. Each plot shows results averaged over 10,000 simulations.

Each simulation uses 20 ads: 10 CPC ads and 10 CPA ads. The CPC ads have \$1 bids and actual priors $\ri = \mathcal{N}(0.001, 0.0001)$ -- Gaussians with mean $0.001$ and standard deviation $0.0001$.  The CPA ads have \$10 bids and actual priors $\ri = \mathcal{N}(0.0001, 0.00001)$, so that their revenues have the same distributions as the CPC ad revenues.

We ran simulations using three possible estimated priors for response rates $\riest$:

\begin{itemize}
\item \textbf{Uniform} -- The prior is uniform over $[0,1]$. This simulates estimating response rates for each ad based on its own performance history, without using a prior based on histories of other ads.
\item \textbf{Approximate} -- The prior is uniform over $[\mu - 4 \sigma, \mu + 4 \sigma]$, where $\mu$ and $\sigma$ are the mean and standard deviation of the actual priors $\ri$. This simulates using a prior based on histories of other similar ads in combination with each ad's own performance history.
\item \textbf{Exact} -- The prior is the actual distribution used to generate response rates, that is, $\mathcal{N}(0.001, 0.0001)$ for CPC ads and $\mathcal{N}(0.0001, 0.00001)$ for CPA ads. This simulates having exact knowledge of response rate distributions, an ideal that does not occur in practice.
\end{itemize}

Figure \ref{est_rev} compares estimated expected revenue for portfolio allocations and single-winner allocations. As $q$ increases, indicating a preference for revenue maximization over variance minimization, estimated expected revenue for portfolio allocations approaches that for single winners, as expected. The estimated expected revenue is shown as a fraction of of ideal expected revenue -- the expected revenue that would be realized using a single winner if response rates were known. With inexact priors, estimated expected revenues can exceed actual ideal revenues, because there is a tendency to select ads with overestimated response rates. (For more on this ``seller's curse'' effect, refer to Bax et al. \cite{baxromero10}.)

\begin{figure*}
\caption{Estimated Expected Revenue}
\begin{center}
\includegraphics[width=100ex,angle=0]{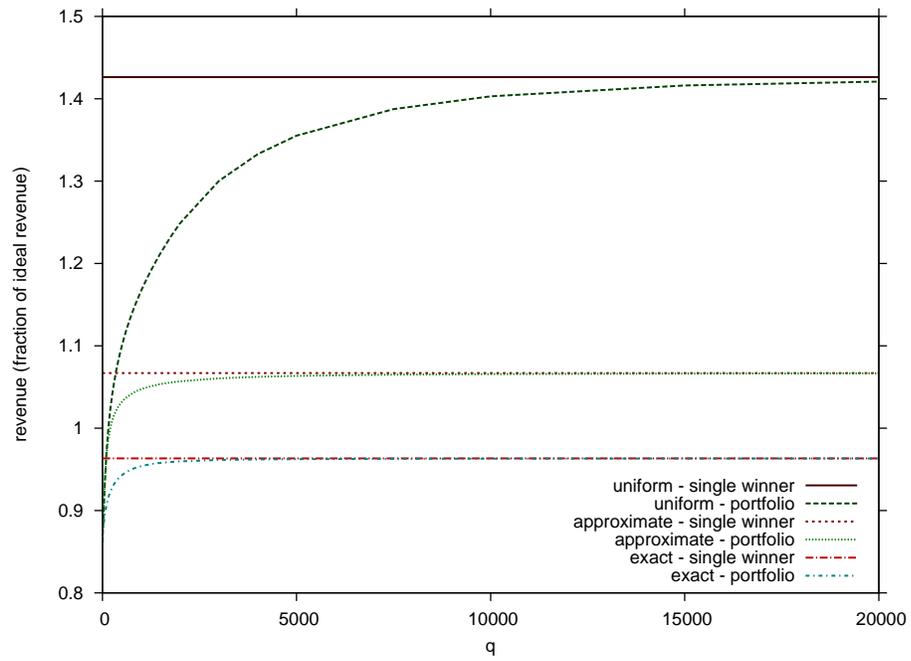}
\end{center}
\label{est_rev}
\end{figure*}

Figures \ref{uni_rev}, \ref{apr_rev}, and \ref{exa_rev} show actual, not estimated, expected revenues. Figures \ref{uni_rev} and \ref{apr_rev} illustrate that using portfolio allocation to control estimated variance can make actual expected revenue higher than for selecting a single ad that maximizes estimated expected revenue. Controlling estimated variance counters the seller's curse of selecting a single winner with overestimated response rate and hence overestimated expected revenue. As in Figure \ref{est_rev}, revenues are shown as fractions of ideal expected revenue. Of course, actual expected revenues do not reach ideal expected revenues, even with exact priors.

\begin{figure*}
  \caption{Actual Expected Revenue -- Uniform Prior}
  \begin{center}
    \includegraphics[width=100ex,angle=0]{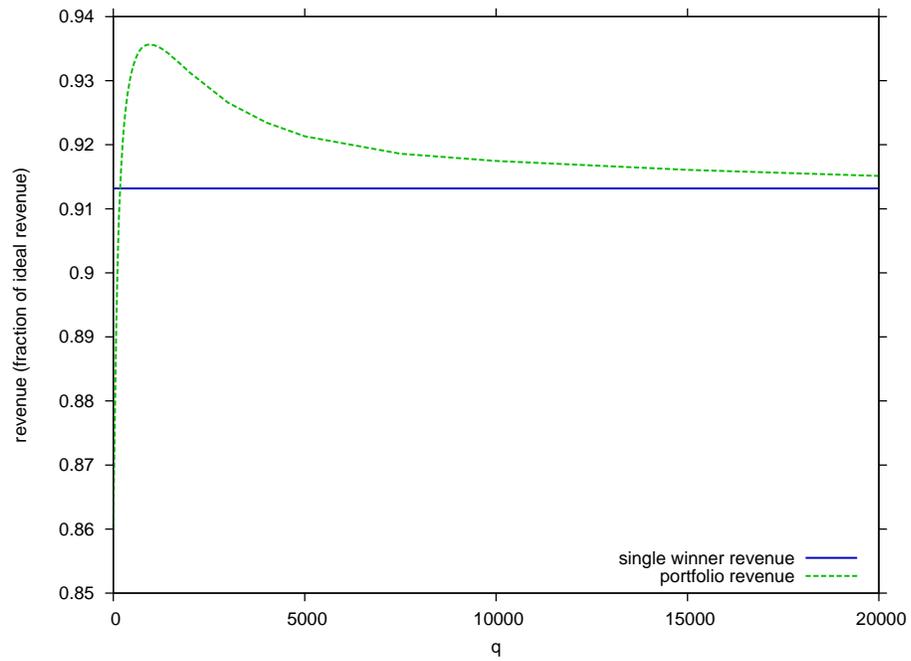}
  \end{center}
  \label{uni_rev}
\end{figure*}

\begin{figure*}
  \caption{Actual Expected Revenue -- Approximate Prior}
  \begin{center}
    \includegraphics[width=100ex,angle=0]{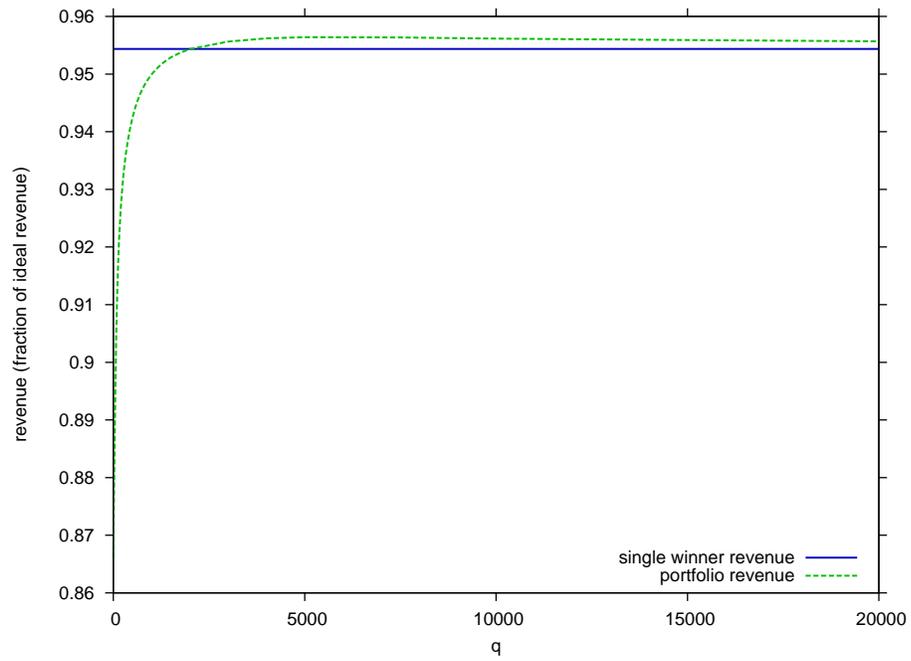}
  \end{center}
  \label{apr_rev}
\end{figure*}

\begin{figure*}
  \caption{Actual Expected Revenue -- Exact Prior}
  \begin{center}
    \includegraphics[width=100ex,angle=0]{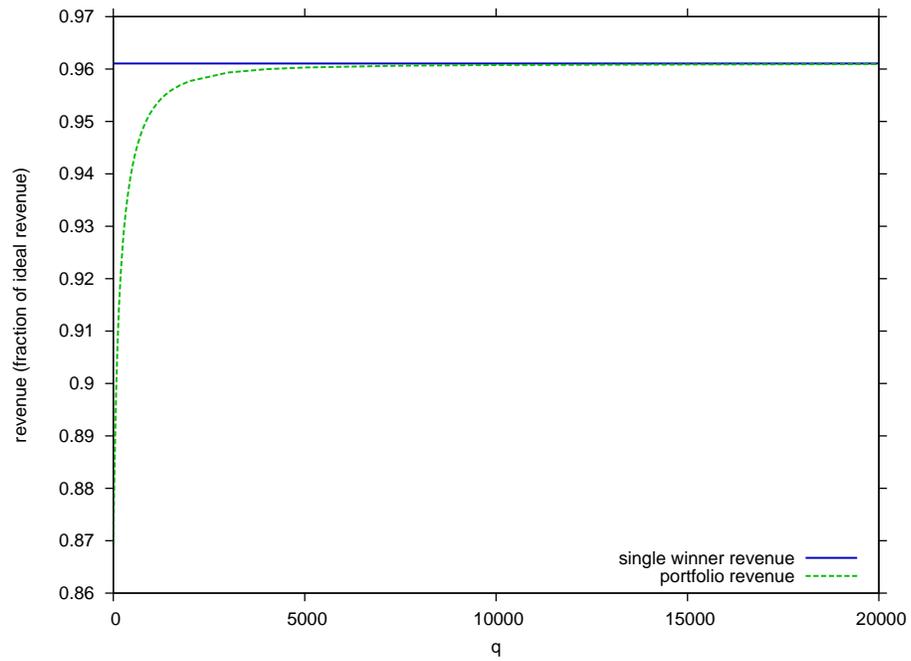}
  \end{center}
  \label{exa_rev}
\end{figure*}

\newpage

These results suggest two steps to improve exchanges for online advertising:
\begin{itemize}
\item Use a Bayesian approach to estimate response probabilities, making the priors as accurate as possible using the best available learning/estimation methods. (Use historical data to categorize ads, deduce the histograms or functional forms of priors, and fit any parameters.)
\item Apply portfolio optimization, experimenting to set $q$ to optimize for a combination of actual expected revenue and, if desired, actual variance of revenue. (For methods to adjust parameters using statistical experiments, refer to Box et al. \cite{box05} or another text on the design of experiments.)
\end{itemize}

\section{Discussion} \label{discussion}
This paper describes a technique to allocate inventory among buyers in online advertising, that mediates a tradeoff between maximizing estimated expected revenue and minimizing estimated variance of revenue (risk). The estimated variance accounts for both uncertainty in estimated response rates and randomness around true response rates. Simulations show that this technique can increase actual expected revenue by preventing the exchange from selecting a single winner based on an overestimate of its payoff.

One direction for future work is to extend the portfolio allocation technique to operate in concert with an explore-exploit method. (For more on explore-exploit methods, also called multi-armed bandit techniques, refer to Bianchi and Lugosi \cite{bianchi06}, Auer et al. \cite{auer02}, Langford et al. \cite{langford02}, Gittins \cite{gittins79a,gittins89}, and Gittins and Jones \cite{gittins79b}.) The technique in this paper can play the role of an exploitation method, offering the side benefit of performing some exploration by allocating ad calls to multiple ads.  One potential approach to extend the portfolio allocation method to include systematic exploration is to add terms to the ad revenue expectations and variances to account for distributions of future revenues due to learning more about response rates. For more on the value of learning, refer to Vermorel and Mohri \cite{vermorel05} and Li et al. \cite{li10}.

Exploration is investing ad calls to learn whether ads' response rates warrant exploiting them by including them in future portfolio allocations. In a sense, exploration to determine an ad's revenue statistics is investing in the option to exploit the ad should it be determined to contribute value to the portfolio. In future research, it would be interesting to explore how this is similar to investing in a call option in a financial market. (For more on options, refer to Hull \cite{hull08}.) In both cases, an upfront investment secures a right to decide whether to make another investment after more information is obtained. In a financial market, the information is revealed over time. In online advertising, the investment buys the information. In both cases, the downside risk is limited to the amount of the initial investment. In a financial market, this occurs when an option is out of the money.

Another direction for future work is to extend methods to accommodate uncertainty in portfolio analysis for financial markets to portfolio analysis for online advertising markets. (For some methods, refer to Jorion \cite{jorion86}, Jobson et al. \cite{jobson79}, and Vasicek \cite{vasicek73}.) It should be useful to apply James-Stein corrections or similar shrinkage methods to estimates of the means, variances, and covariances of revenue distributions for ads. (For more on shrinkage methods, refer to Bock \cite{bock75}, Brown \cite{brown66}, Stein \cite{stein55}, and James and Stein \cite{james61}.)

The field of robust optimization focuses on optimization under uncertainty. Some robust optimization approaches address \textit{strict uncertainty} (see Sniedovich \cite{sniedovich07}), where the probabilities of possible outcomes are completely unknown. Others address less uncertain problems, where the distribution over possible outcomes is unknown but restricted to some set of distributions (see Ben-Tal and Nemirovski \cite{bental98}, Ben-Haim \cite{benhaim05}, and Chen et al. \cite{chen07}).  In this paper, we began by examining the effect on \textit{risk} (as defined by French \cite{french88}) of drawing a distribution over outcomes (corresponding to $\mathbf{S}$) from a distribution over distributions (corresponding to $\mathbf{R}$). Then we used simulations to explore the effect of having imperfect information about the distribution over distributions. This introduces a form of uncertainty beyond risk, but not as severe as strict uncertainty. In the future, it would be interesting to apply the methods of robust optimization to ad allocation problems with uncertainty about the parameters (such as the number of ad calls available) as well as the payoffs. For more on robust optimization for portfolio problems, refer to Fabozzi et al. \cite{fabozzi07} and Goldfarb and Iyengar \cite{goldfarb03}.

Finally, it would be interesting to analyze questions about bidding behavior under portfolio allocations. For example, could some advertisers benefit by switching price types from CPC to CPA or vice versa to adjust the variance of their ads' revenues under portfolio allocation? Under a mixed allocation, there are multiple winners, so what is the second price? One answer to the latter question is VCG (Vickrey-Clarke-Groves \cite{groves73,clarke71,vickrey62,vickrey61}) pricing for portfolio allocations, as outlined by Li et al. in \cite{vcg_portfolio}. But there may be other approaches that maintain incentives to bid truthfully and increase revenue.

\appendix
\section{Variance of Allocation Payoff} \label{payoff_variance}
\newtheorem{theorem}{Theorem}
\begin{theorem}
\be
Var_{\SX} \rksx
\ee
\be
= \sum_{i=1}^{n} \sum_{j=1}^{n} k_i k_j Cov_{S_i,S_j} \left[ E_{X_i} X_i(S_i), E_{X_j} X_j(S_j) \right] 
\ee
\be
+ \sum_{i=1}^{n} k_i E_{S_i} Var_{X_i} \xsi.
\ee
\noindent
\end{theorem}

\begin{proof}
Use the well-known equality \cite{feller68} for variance: $Var X = E X^2 - (E X)^2$:

\be
Var_{\SX} \rksx = E_{\SX} \rksx^2 - \left[ E_{\SX} \rksx \right]^2.
\ee
\noindent
For the first term, separate expectations for $\SSS$ and $\XX$, and apply the equality $E X^2 = Var X + (E X)^2$:

\belabel
= E_{\SSS} \left[ Var_{\XX}  \rksx^2 \right] 
\ee
\be
+ E_{\SSS} \left[E_{\XX} \rksx\right]^2 
\ee
\be
- \left[E_{\SX} \rksx\right]^2. \label{three_terms}
\eelabel
\noindent

Now expand the three terms one at a time. For the first term, use the definition of $\rksx$:

\be
E_{\SSS} \left[ Var_{\XX}  \rksx^2 \right] 
\ee
\be
= E_{\SSS} Var_{\XX} \sum_{i=1}^{n} \left[\sum_{h=k_1 + \ldots + k_{i-1}}^{k_1 + \ldots + k_i} \xhi \right].
\ee
\noindent
Since payoffs are i.i.d. with respect to $\XX$,

\be
E_{\SSS} \left[ Var_{\XX}  \rksx^2 \right]  
\ee
\be
= E_{\SSS} \sum_{i=1}^{n} k_i Var_{X_i} \xsi 
\ee
\be
= \sum_{i=1}^{n} k_i E_{S_i} \left[ Var_{X_i} \xsi \right].
\ee
\noindent
This is the last term on the RHS of the equation in the statement of the theorem.

Next, expand the second term of Equation (\ref{three_terms}). Use the definition of $\rksx$.

\be
E_{\SSS} \left[E_{\XX} \rksx\right]^2 
\ee
\be
= E_{\SSS} \left[\left(E_{\XX} \sum_{i=1}^{n} \left[\sum_{h=k_1 + \ldots + k_{i-1}}^{k_1 + \ldots + k_i} \xhi \right] \right) \right.
\ee
\be
\left. \left(E_{\XX} \sum_{j=1}^{n} \left[ \sum_{g=k_1 + \ldots + k_{j-1}}^{k_1 + \ldots + k_j} X_{gj} (S_j)\right]\right)\right].
\ee
\noindent
Since payoffs are i.i.d. with respect to $\XX$,

\be
= E_{\SSS} \left[\sum_{i=1}^{n} k_i E_{X_i} \xsi\right] \left[\sum_{j=1}^{n} k_j E_{X_j} X_j (S_j)\right].
\ee
\noindent
Distribute $E_{\SSS}$ and multiply the sums term-by-term.

\belabel
= \sum_{i=1}^{n} \sum_{j=1}^{n} k_i k_j E_{\SSS} \left(E_{X_i} \xsi \cdot E_{X_j} X_j (S_j)\right). \label{canceler}
\eelabel
\noindent
Now expand the third term of Equation (\ref{three_terms}). Substitute in the expectation of $\rksx$ from Equation (\ref{expectation}).

\be
- \left[E_{\SX} \rksx\right]^2 = - \left[\sum_{i=1}^{n} k_i E_{S_i, X_i} \xsi\right]^2
\ee
\noindent
Expand the square.

\be
= - \sum_{i=1}^{n} \sum_{j=1}^{n} k_i k_j E_{S_i, X_i} \xsi E_{S_j, X_j} X_j (S_j).
\ee
\noindent
Apply the equality for covariance \cite{feller68}: $Cov (X,Y) = E XY - (E X) (E Y)$.

\be
= - \sum_{i=1}^{n} \sum_{j=1}^{n} k_i k_j \left[-E_{S_i, S_j} \left(E_{X_i} \xsi \cdot E_{X_j} X_j (S_j)\right) \right.
\ee
\be
\left. + Cov_{S_i, S_j} \left(E_{X_i} \xsi, E_{X_j} X_j (S_j)\right)\right]
\ee
\noindent
Carry through the sign and separate the expectation and covariance terms.

\be
= - \sum_{i=1}^{n} \sum_{j=1}^{n} k_i k_j E_{S_i, S_j} \left(E_{X_i} \xsi \cdot E_{X_j} X_j (S_j)\right) 
\ee
\be
+ \sum_{i=1}^{n} \sum_{j=1}^{n} k_i k_j Cov_{S_i, S_j} \left(E_{X_i} \xsi, E_{X_j} X_j (S_j)\right).
\ee
\noindent
The first term cancels Equation (\ref{canceler}). The second term completes the RHS of the statement of the theorem.
\end{proof}

\bibliography{bax}
\bibliographystyle{abbrv}

\end{document}